\documentclass[10pt]{article}
\pdfoutput=1

\usepackage[hscale=0.8, vscale=0.8]{geometry}

\usepackage[compact]{titlesec}

\usepackage{microtype} %
\newcommand{\shrinkitems}{\setlength{\itemsep}{0ex}} %
\newcommand{\shrinktable}{\setlength{\abovecaptionskip}{1ex}\setlength{\belowcaptionskip}{-1ex}} %

\usepackage[utf8]{inputenc}
\usepackage[T1]{fontenc}

\usepackage[usenames]{xcolor}

\newcommand{\ttlpar}[1]{\noindent\textbf{#1}}

\usepackage{bold-extra}
\usepackage[final]{listings}
\usepackage{courier}
\usepackage{booktabs}
\usepackage{multirow}
\usepackage[pdfborder={0 0 0}]{hyperref}
\usepackage{balance}
\usepackage[bf, small]{caption}
\usepackage{bm}

\usepackage{enumerate}
\usepackage{amsmath}
\usepackage[amsmath, thmmarks]{ntheorem}

\theorembodyfont{\upshape}
\theoremheaderfont{\sc}
\newtheorem{theorem}{Theorem}[section]
\newtheorem{lemma}[theorem]{Lemma}
\newtheorem{definition}[theorem]{Definition}
\newtheorem{observation}[theorem]{Remark}
\theoremheaderfont{\it}
\theoremstyle{nonumberplain}
\theoremseparator{}
\theoremsymbol{\rule{1ex}{1ex}}
\newtheorem{proof}{Proof}

\usepackage{tikz}
\usetikzlibrary{positioning}
\usetikzlibrary{shapes}

\DeclareMathOperator{\Access}{Access}
\DeclareMathOperator{\Rank}{Rank}
\DeclareMathOperator{\Select}{Select}

\DeclareMathOperator{\Append}{Append}
\DeclareMathOperator{\Insert}{Insert}
\DeclareMathOperator{\Delete}{Delete}
\DeclareMathOperator{\Init}{Init}
\DeclareMathOperator{\Pred}{Predecessor}

\DeclareMathOperator{\RankPrefix}{RankPrefix}
\DeclareMathOperator{\SelectPrefix}{SelectPrefix}

\DeclareMathOperator{\WT}{WT}
\DeclareMathOperator{\LT}{LT}
\DeclareMathOperator{\LB}{LB}
\DeclareMathOperator{\PT}{PT}
\DeclareMathOperator{\Prob}{Prob}
\DeclareMathOperator{\polylog}{polylog}
\DeclareMathOperator{\RangeCount}{RangeCount}

\newcommand{\pos}{\mathrm{pos}}
\newcommand{\occidx}{\mathrm{idx}}
\newcommand{\Cop}{\mathcal{C}_\mathrm{op}}

\newcommand{\bit}[1]{\text{\bf\texttt{#1}}}
\newcommand{\bitnobf}[1]{\text{\texttt{#1}}}
\newcommand{\bitzero}{\bit{0}}
\newcommand{\bitone}{\bit{1}}
\newcommand{\Sset}{S_\text{set}}
\newcommand{\logbinom}{\mathcal{B}}

\begin{document}

\title{The Wavelet Trie: Maintaining an\\Indexed Sequence of Strings
  in Compressed Space\footnote{Work partially supported by MIUR of Italy
    under project AlgoDEEP prot.\mbox{} 2008TFBWL4.}}

\author{
  Roberto Grossi \\ \small Università di Pisa \\ \small \texttt{grossi@di.unipi.it}
  \and
  Giuseppe Ottaviano\footnote{Part of the work done while the author was an intern at Microsoft Research, Cambridge.} \\ \small Università di Pisa \\ \small \texttt{ottavian@di.unipi.it}
}

\date{}

\maketitle

\begin{abstract}
  An \emph{indexed sequence of strings} is a data structure for
  storing a \emph{string sequence} that supports random access,
  searching, range counting and analytics operations, both for exact
  matches and prefix search. String sequences lie at the core of
  column-oriented databases, log processing, and other storage and
  query tasks. In these applications each string can appear several
  times and the order of the strings in the sequence is relevant.
  The prefix structure of the strings is relevant as well:
  common prefixes are sought in strings to extract interesting
  features from the sequence. Moreover, space-efficiency is highly
  desirable as it translates directly into higher performance, since
  more data can fit in fast memory.

  We introduce and study the problem of \emph{compressed indexed
    sequence of strings}, representing indexed sequences of strings in
  nearly-optimal compressed space, both in the static and dynamic
  settings, while preserving provably good performance for the
  supported operations.

  We present a new data structure for this problem, the \emph{Wavelet
    Trie}, which combines the classical Patricia Trie with the Wavelet
  Tree, a succinct data structure for storing a compressed
  sequence. The resulting Wavelet Trie smoothly adapts to a sequence
  of strings that changes over time. It improves on the
  state-of-the-art compressed data structures by supporting a dynamic
  alphabet (i.e.\mbox{} the set of distinct strings) and prefix
  queries, both crucial requirements in the aforementioned
  applications, and on traditional indexes by reducing space occupancy
  to close to the entropy of the sequence.
\end{abstract}

\section{Introduction}

Many problems in databases and information retrieval ultimately reduce
to storing and indexing sequences of strings. Column-oriented
databases represent relations by storing individually each column as
a sequence; if each column is indexed, efficient operations on the
relations are possible. XML databases, taxonomies, and word tries are represented
as labeled trees, and each tree can be mapped to the sequence of its
labels in a specific order; %
indexed operations on the sequence enable fast tree navigation. In
data analytics query logs and access logs are simply sequences of
strings; aggregate queries and counting queries can be performed
efficiently with specific indexes. Textual document search is
essentially the problem of representing a text as the sequence of its
words, and queries locate the occurrences of given words in the
text. Even the storage of non-string (for example, numeric) data can
be often reduced to the storage of strings, as usually the values
can be binarized in a natural way.

\smallskip
\ttlpar{Indexed sequence of strings.}
An \emph{indexed sequence of strings} is a data structure for storing
a \emph{string sequence} that supports random access, searching, range
counting and analytics operations, both for exact matches and prefix
search. Each string can appear several times and the order of the
strings in the sequence is relevant. For a sequence $S \equiv \langle
s_0, \ldots, s_{n-1} \rangle$ of strings, the primitive operations
are:
\begin{itemize}
  \shrinkitems
\item $\Access(\pos)$: retrieve string $s_{\pos}$, where $0 \leq \pos
  < n$.
\item $\Rank(s, \pos)$: count the number of occurrences of string $s$
  in $\langle s_0, \ldots, s_{\pos-1} \rangle$.
\item $\Select(s, \occidx)$: find the position of the $\occidx$-th occurrence of
  $s$ in $\langle s_0, \ldots, s_{n-1} \rangle$.
\end{itemize}
By composing these three primitives it is possible to implement other
powerful index operations.  For example, functionality similar to inverted lists can be
easily formulated in terms of $\Select$. These primitives can be
extended to prefixes. 
\begin{itemize}
\shrinkitems
\item $\RankPrefix(p, \pos)$: count the number of strings in
  $\langle s_0, \ldots, s_{\pos-1} \rangle$ that have prefix $p$.
\item $\SelectPrefix(p, \occidx)$: find the position of the $\occidx$-th string in
  $\langle s_0, \ldots, s_{n-1} \rangle$ that has prefix $p$.
\end{itemize}
Prefix search operations can be easily formulated in terms of
$\SelectPrefix$.  Other useful operations are range counting and
analytics operations, where the above primitives are generalized to a
range $[\pos', \pos)$ of positions, hence to $\langle s_{\pos'},
\ldots, s_{\pos-1} \rangle$.  In this way statistically interesting
(e.g.\mbox{} frequent) strings in the given range $[\pos', \pos)$ and
having a given prefix~$p$ can be quickly discovered (see
Section~\ref{sec:query-algor-wavel} for further operations). 

The
sequence~$S$ can change over time by defining the following operations, 
for any arbitrary string $s$ (which could be \emph{previously unseen}).
\begin{itemize}\shrinkitems
\item $\Insert(s, \pos)$: update the sequence $S$ as $S' \equiv
  \langle s_0, \ldots, s_{\pos-1}, s, s_{\pos}, \ldots, s_{n-1}
  \rangle$ by inserting $s$ immediately before $s_{\pos}$.

\item $\Append(s)$: update the sequence $S$ as $S' \equiv \langle
  s_0, \ldots, s_{n-1}, s \rangle$ by appending $s$ at the end.

\item $\Delete(\pos)$: update the sequence $S$ as $S' \equiv
  \langle s_0, \ldots, s_{\pos-1}, s_{\pos+1}, \ldots, s_{n-1}
  \rangle$ by deleting $s_{\pos}$.

\end{itemize}

\ttlpar{Motivation.}
String sequences lie at the core of column-oriented databases, log
processing, and other storage and query tasks. The prefix operations
supported by an indexed sequence of strings arise in many
contexts. Here we give a few examples to show that the proposed
problem is quite natural. In data analytics for query logs and access
logs, the sequence order is the \emph{time} order, so that a range of
positions $[\pos', \pos)$ corresponds to a given time frame.
The accessed URLs,
paths (filesystem, network, \dots) or any kind of hierarchical references are chronologically stored as a
sequence $\langle s_0, \ldots, s_{n-1} \rangle$ of strings, and a common
prefix denotes a common domain or a common folder for the given time
frame: we can retrieve access statistics using $\RankPrefix$ and
report the corresponding items by iterating $\SelectPrefix$
(e.g.\mbox{} ``what has been the most accessed domain during winter
vacation?''). This has a wide array of applications, from intrusion
detection and website optimization to database storage of telephone
calls.  Another interesting example arises in web graphs and social
networks, where a \emph{binary relation} is stored as a graph among
the entities, so that each edge is conceptually a pair of URLs or
hierarchical references (URIs). Edges can change over time, so we can
report what changed in the adjacency list of a given vertex in a given
time frame, allowing us to produce snapshots on the fly (e.g.\mbox{}
``how did friendship links change in that social network during winter
vacation?'').  In the above applications the many strings involved
require a suitable compressed format. Space-efficiency is highly
desirable as it translates directly into higher performance, since more
data can fit in fast memory.

\smallskip
\ttlpar{Compressed indexed sequence of strings.}
We introduce and study the problem of \emph{compressed indexed
  sequence of strings} representing indexed sequences of strings in
nearly-optimal compressed space, both in the static and dynamic
settings, while preserving provably good performance for the supported
operations. 

Traditionally, indexed sequences are stored by representing the
sequence explicitly and indexing it using auxiliary data structures,
such as B-Trees, Hash Indexes, Bitmap Indexes.  These data structures
have excellent performance and both external and cache-oblivious
variants are well studied \cite{vitter2008algorithms}. Space efficiency is however sacrificed: the
total occupancy is several times the space of the sequence alone. In a
latency constrained world where more and more data have to be kept in
internal memory, this is not feasible anymore.

The field of succinct and compressed data structures comes to aid:
there is a vast literature about compressed storage of sequences,
under the name of $\Rank$/$\Select$ sequences \cite{jacobson89}. 
The existing $\Rank$/$\Select$ data structures, however, assume that
the alphabet from which the sequences are drawn is \emph{integer} and
\emph{contiguous}, i.e.\mbox{} each element of the sequence is just a
symbol in $\{1, \dots, \sigma\}$. Non-integer or non-contiguous
alphabets need to be mapped first to an integer range.  Letting
$\Sset$ denote the set of distinct strings in the sequence $S \equiv
\langle s_0, \ldots, s_{n-1} \rangle$, the representation of~$S$ as a
sequence of $n$ integers in $\{1, \dots, |\Sset|\}$ requires to map
each $s_i$ to its corresponding integer, thus introducing at least two
issues: $(a)$~once the mapping is computed, it cannot be changed,
which means that in dynamic operations the alphabet must be known in
advance; $(b)$~for string data the string structure is lost, hence no
prefix operations can be supported.  Issue~$(a)$ in particular rules
out applications in database storage, as the set of values of a
column (or even its cardinality) is very rarely known in
advance; similarly in text indexing a new document can contain
unseen words; in URL sequences, new URLs can be created at any moment.

\smallskip
\ttlpar{Wavelet Trie.}
We introduce a new data structure, the \emph{Wavelet Trie}, that
overcomes the previously mentioned issues. The Wavelet Trie is a
generalization for string sequences~$S$ of the Wavelet Tree
\cite{DBLP:conf/soda/GrossiGV03}, a compressed data structure for
sequences, where the shape of the tree is induced from the structure of
the string set $\Sset$ as in the Patricia Trie
\cite{DBLP:journals/jacm/Morrison68}. This enables efficient prefix
operations and the ability to grow or shrink the alphabet as values
are inserted or removed. We first present a static version of the
Wavelet Trie in Section~\ref{sec:wavelet-trie}. 
We then give an append-only dynamic version of the
Wavelet Trie, meaning that elements can be inserted only at the end---the
typical scenario of query logs and access logs---and a fully dynamic
version that is useful for database applications (see
Section~\ref{sec:dynamic-wavelet-tries}).

Our time bounds are reported in Table~\ref{tab:bounds}, and some
comments are in order. Recall that $S$ denotes the input sequence of strings
stored in the Wavelet Trie, and $\Sset$ is the set of distinct strings in
$S$.  For a string $s$ to be queried, let $h_s$ denote the number of
nodes traversed in the binary Patricia Tree storing $\Sset$ when $s$
is searched for. Observe that $h_s \leq |s| \log |\Sigma|$, where
$\Sigma$ is the alphabet of symbols from which $s$ is drawn, and $|s|
\log |\Sigma|$ is the length in bits of $s$ (while $|s|$ denotes its
number of symbols as usual). The cost for the queries on the static
and append-only versions of the Wavelet Trie is $O(|s| + h_s)$ time,
which is the \emph{same} cost as searching in the binary Patricia
Trie. Surprisingly, the cost of appending $s$ to $S$ is still $O(|s| +
h_s)$ time, which means that compressing and indexing a sequential log
on the fly is very efficient. The cost of the operations for the fully
dynamic version are also competitive, without the need of knowing the
alphabet in advance. This answers positively a question posed in
\cite{navarrodynamicbitvectors09} and
\cite{DBLP:journals/talg/MakinenN08}.

\begin{table}[htbp]
  \shrinktable
  \centering
  {
    \small
    \setlength{\tabcolsep}{1ex}
\begin{tabular}{lccccc}\toprule
  &  \textbf{Query} & \textbf{Append} & \textbf{Insert} & \textbf{Delete} & \textbf{Space (in bits)} \\
  \midrule
  {Static} & $O(|s| + h_s)$ & -- & -- & -- & $\LB + o(\tilde h n)$ \\
  {Append-only} & $O(|s| + h_s)$ & $O(|s| + h_s)$ & -- & -- & $\LB + \PT + o(\tilde h n)$ \\
  {Fully-dynamic} & $O(|s| + h_s\log n)$ & $O(|s| + h_s\log n)$ & $O(|s| + h_s\log n)$ & $O(|s| + h_s\log n)$${}^\dagger$ & $\LB + \PT + O(nH_0)$ \\
  \bottomrule
\end{tabular}
  }
  \caption{Bounds for the Wavelet Trie. 
    \emph{Query} is the cost of Access, Rank(Prefix), Select(Prefix),
    $\LB$ is the information theoretic lower bound $\LT + nH_0$
    (Sect.~\ref{sec:wavelet-trie}), and
    $\PT$ the space taken by the dynamic Patricia Trie
    (Sect.~\ref{sec:dynamic-wavelet-tries}). ${}^\dagger$Note that deletion may
    take $O(\hat \ell + h_s \log n)$ time when deleting the last occurrence of a string, where $\hat \ell$ is the length of the longest string in
    $\Sset$. }
  \label{tab:bounds}
\end{table}

All versions are nearly optimal in space as shown in
Table~\ref{tab:bounds}. In particular, the lower bound $\LB(S)$ for
storing an indexed sequence of strings can be derived from the lower
bound $\LT(\Sset)$ for storing $\Sset$ given in \cite{pods08} plus
Shannon classical zero-order entropy bound $n H_0(S)$ for storing $S$ 
as a sequence of symbols. The static version uses an additional number of
bits that is just a lower order term $o({\tilde h} n)$, where $\tilde
h$ is the average height of the Wavelet Tree
(Definition~\ref{def:avgheight}). The append-only version only adds
$\PT(\Sset) = O(|\Sset| \log n)$ bits for keeping $O(|\Sset|)$ pointers
to the dynamically allocated memory (assuming that we do not have
control on the memory allocator on the machine). The fully dynamic
version has a redundancy of $O(n H_0(S))$ bits.

\smallskip
\ttlpar{Results.}
Summing up the above contributions: we address a new problem on
sequences of strings that is meaningful in real-life applications; we
introduce a new compressed data structure, the Wavelet Trie, and
analyze its nearly optimal space; we show that the supported
operations are competitive with those of uncompressed data structures,
both in the static and dynamic setting. We have further findings in
this paper.  In case the prefix operations are not needed (for example
when the values are numeric), we show in Section~\ref{sec:probabilistic-balanced-dynamic-wt} 
how to use a Wavelet Trie
to maintain a probabilistically balanced Wavelet Tree, hence
guaranteeing access times logarithmic in the alphabet size. Again, the
alphabet does not need to be known in advance. We also present an
append-only compressed bitvector that supports constant-time
$\Rank$, $\Select$, and $\Append$ in nearly optimal
space. We use this bitvector in the append-only Wavelet Trie.

\smallskip
\ttlpar{Related work.}
While there has been extensive work on compressed representations for
\emph{sets} of strings, to the best of our knowledge the problem of
representing \emph{sequences} of strings has not been studied. Indexed
sequences of strings are usually stored in the following ways:
\emph{(1)}~by mapping the strings to integers through a dictionary,
the problem is reduced to the storage of a sequence of integers;
\emph{(2)}~by concatenating the strings with a separator, and
compressing and full-text indexing the obtained string; \emph{(3)}~by
storing the concatenation $(s_i, i)$ in a string dictionary such as a
B-Tree.

The approach in \emph{(1)}, used implicitly in
\cite{CNspire08.1,DBLP:journals/jacm/FerraginaLMM09} and most
literature about $\Rank$/$\Select$ sequences, sacrifices the ability
to perform prefix queries. If the mapping preserves the lexicographic
ordering, prefixes are mapped to contiguous ranges; this enables some
prefix operations, by exploiting the two-dimensional nature of the
Wavelet Tree: $\RankPrefix$ can be reduced to the $\RangeCount$
operation described in \cite{MakinenN06}. To the best of our
knowledge, however, even with a lexicographic mapping there is no way
to support efficiently $\SelectPrefix$. More importantly, in the
dynamic setting it is not possible to change the alphabet (the
underlying string set~$\Sset$) without rebuilding the tree, as
previously discussed.

The approach in \emph{(2)}, called \emph{Dynamic Text
  Collection} in \cite{DBLP:journals/talg/MakinenN08}, although it
allows for potentially more powerful operations, is both slower, because
it needs a search in the compressed text index,
and less space-efficient,
as it only compresses according to the $k$-order entropy of the
string, failing to exploit the redundancy given by repeated strings.

The approach
in \emph{(3)}, used often in databases to implement indexes, only
supports $\Select$, while another copy of the sequence is still needed
to support $\Access$, and it does not support $\Rank$. Furthermore, it
offers little or no guaranteed compression ratio.

\section{Preliminaries}\label{sec:preliminaries}

\ttlpar{Information-theoretic lower bounds.}
We assume that all the logarithms are in base $2$, and that the word size is
$w \geq \log n$ bits. 
Let $s = c_1 \dots c_n \in \Sigma^*$ be a sequence of length $|s| =
n$, drawn from an alphabet $\Sigma$.  The \emph{binary representation}
of $s$ is a binary sequence of $n \lceil \log |\Sigma| \rceil$ bits,
where each symbol $c_i$ is replaced by the $\lceil \log |\Sigma|
\rceil$ bits encoding it. The \emph{zero-order empirical entropy} of
$s$ is defined as $H_0(s) = -\sum_{c \in \Sigma} \frac{n_c}{n}
\log\frac{n_c}{n}$, where $n_c$ is the number of occurrences of symbol
$c$ in $s$. Note that $n H_0(s) \leq n \log |\Sigma|$ is a lower bound
on the bits needed to represent $s$ with an encoder that does not exploit context information.  If
$s$ is a binary sequence with $\Sigma = \{ \bitzero, \bitone \}$ and $p$
is the fraction of $\bitone$s in $s$, we can rewrite the entropy as
$H_0(s) = - p \log p - (1 - p) \log (1 - p)$, which we also denote by
$H(p)$.
We use $\logbinom(m, n)$ as a shorthand for $\lceil \log {n \choose m}
\rceil$, the information-theoretic lower bound in bits for storing a
set of $m$ elements drawn from an universe of size $n$. We implicitly
make extensive use of the bounds $\logbinom(m, n) \leq
nH(\frac{m}{n}) + O(1)$, and $\logbinom(m, n) \leq m \log(\frac{n}{m}) +
O(m)$.

\smallskip
\ttlpar{Bitvectors and FIDs.}
Binary sequences, i.e. $\Sigma = \{\bitzero, \bitone\}$, are also
called \emph{bitvectors}, and data structures that encode a bitvector
while supporting $\Access$/$\Rank$/$\Select$ are also called
\emph{Fully Indexed Dictionaries}, or \emph{FID}s \cite{RRR07}. The
representation of \cite{RRR07}, referred to as RRR, can encode a
bitvector with $n$ bits, of which $m$ $\bitone$s, in $\logbinom(m, n) + O((n
\log \log n)/\log n)$ bits, while supporting all the operations in
constant time. 

\smallskip
\ttlpar{Wavelet Trees.}
The Wavelet Tree, introduced in \cite{DBLP:conf/soda/GrossiGV03}, is
the first data structure to extend $\Rank$/$\Select$ operations from
bitvectors to sequences on an arbitrary alphabet $\Sigma$, while
keeping the sequence compressed. Wavelet Trees reduce the problem to
the storage of a set of $|\Sigma|-1$ bitvectors organized in a tree
structure.

The alphabet is recursively partitioned in two subsets, until each
subset is a singleton (hence the leaves are in one-to-one correspondence with the
symbols of $\Sigma$).
The bitvector $\beta$ at the root has one bit for each element of the
sequence, where $\beta_i$ is $\bitzero$/$\bitone$ if the $i$-th
element belongs to the left/right subset of the alphabet. The sequence
is then projected on the two subsets, obtaining two subsequences, and
the process is repeated on the left and right
subtrees. An example is shown in Figure~\ref{fig:wavelettree}.

Note that the $\bitzero$s of one node are in one-to-one
correspondence with the bits of the left node, while the $\bitone$s
are in correspondence with the bits of the right node, and the
correspondence is given downwards by $\Rank$ and upwards by
$\Select$. Thanks to this mapping, it is possible to perform $\Access$
and $\Rank$ by traversing the tree top-down, and $\Select$ by
traversing it bottom-up.

By using RRR bitvectors, the space is $n H_0(S) + o(n \log
|\Sigma|)$ bits, while operations take $O(\log |\Sigma|)$ time.

\smallskip

\ttlpar{Patricia Tries.} 
The
\emph{Patricia Trie} \cite{DBLP:journals/jacm/Morrison68} (or
\emph{compacted binary trie}) of a non-empty prefix-free set of binary
strings is a binary tree built recursively as follows.  $(i)$~The
Patricia Trie of a single string is a node labeled with the string.
$(ii)$~For a nonempty string set $\mathcal{S}$, let $\alpha$ be the
longest common prefix of $\mathcal{S}$ (possibly the empty
string). Let $\mathcal{S}_b = \{\gamma | \alpha b\gamma \in
\mathcal{S}\}$ for $b \in \{\bitzero, \bitone\}$.  Then the Patricia
trie of $\mathcal{S}$ is the tree whose root is labeled with $\alpha$
and whose children (respectively labeled with $\bitzero$ and
$\bitone$) are the Patricia Tries of the sets $\mathcal{S}_\bitzero$
and $\mathcal{S}_\bitone$.
Unless otherwise specified, we use \emph{trie} to indicate a
\emph{Patricia Trie}, and we focus on binary strings.

\section{The Wavelet Trie}\label{sec:wavelet-trie}

\begin{figure}[t]
  \centering
  \begin{minipage}[t]{7.5cm}
    \begin{center}
\begin{tikzpicture}[
  scale=0.8, transform shape,
  intnode/.style={rectangle, draw=black, fill=none, thick, anchor=north},
  leaf/.style={rectangle, draw=black, fill=none, color=black, minimum width=1.5cm, minimum height=0.7cm, thick, anchor=north},
  child anchor=north,
  edge from parent/.style={draw=black, thick},
  level distance=1.1cm,
  level 1/.style={sibling distance=4.5cm},
  level 2/.style={sibling distance=2.2cm},
  level 3/.style={sibling distance=1.8cm}
]
  \node [intnode] {
    $\begin{array}{c}
      \bitnobf{abracadabra}\\
      \bit{00101010010}
    \end{array}$
  }
  child {
    node [intnode] {
      $\begin{array}{c}
        \bitnobf{abaaaba}\\
        \bit{0100010}
      \end{array}$
    }
    child {
      node [leaf] {\bitnobf{a}}
      edge from parent node[above left, anchor=340] {$\{\bitnobf{a}\}$}
    }
    child {
      node [leaf] {\bitnobf{b}}
      edge from parent node[above right, anchor=200] {$\{\bitnobf{b}\}$}
    }
    edge from parent node[above left, anchor=340] {$\{\bitnobf{a}, \bitnobf{b}\}$}
  } 
  child {
    node [intnode] {
      $\begin{array}{c}
        \bitnobf{rcdr}\\
        \bit{1011}
      \end{array}$
    }
    child {
      node [leaf] {\bitnobf{c}}
      edge from parent node[above left, anchor=340] {$\{\bitnobf{c}\}$}
    }
    child {
      node [intnode] {
        $\begin{array}{c}
          \bitnobf{rdr}\\
          \bit{101}
        \end{array}$
      }
      child {
        node [leaf] {\bitnobf{d}}
        edge from parent node[above left, anchor=340] {$\{\bitnobf{d}\}$}
      }
      child {
        node [leaf] {\bitnobf{r}}
        edge from parent node[above right, anchor=200] {$\{\bitnobf{r}\}$}
      }
      edge from parent node[above right, anchor=200] {$\{\bitnobf{d}, \bitnobf{r}\}$}
    } 
    edge from parent node[above right, anchor=200] {$\{\bitnobf{c}, \bitnobf{d}, \bitnobf{r}\}$}
  } 
  ;
\end{tikzpicture}
      \caption{A Wavelet Tree for the input sequence \bit{abracadabra} from
        the alphabet $\{\bit{a}, \bit{b}, \bit{c}, \bit{d},
        \bit{r}\}$.}
      \label{fig:wavelettree}
    \end{center}
  \end{minipage}
  \hspace{0.5cm}
  \begin{minipage}[t]{7.5cm}
    \begin{center}
\begin{tikzpicture}[
  scale=0.8, transform shape,
  intnode/.style={rectangle, draw=black, fill=none, thick, anchor=north},
  leaf/.style={rectangle, draw=black, fill=none, color=black, minimum width=1.5cm, minimum height=0.7cm, thick, anchor=north},
  child anchor=north,
  edge from parent/.style={draw=black, thick},
  level distance=1cm,
  level 1/.style={sibling distance=3.8cm},
  level 2/.style={sibling distance=2.8cm},
  level 3/.style={sibling distance=2cm}
]
  \node [intnode] {
    $\begin{array}{c}
      \alpha: \bit{0}\\
      \beta: \bit{0010101}\\
    \end{array}$
  }
  child {
    node [intnode] {
      $\begin{array}{c}
        \alpha: \varepsilon\\
        \beta: \bit{0111}\\
      \end{array}$
    }
    child {
      node [leaf] {$\alpha: \bit{1}$}
    }
    child {
      node [intnode] {
        $\begin{array}{c}
          \alpha: \varepsilon\\
          \beta: \bit{100}\\
        \end{array}$
      }
      child {
        node [leaf] {$\alpha: \varepsilon$}
      }
      child {
        node [leaf] {$\alpha: \bit{00}$}
      }
    } 
  } 
  child {
    node [leaf] {$\alpha: \bit{00}$}
  } 
  ;
\end{tikzpicture}
      \caption{The Wavelet Trie of the sequence of strings
        $\langle \bit{0001}, \bit{0011}, \bit{0100}, \bit{00100}, \bit{0100},
        \bit{00100}, \bit{0100} \rangle$.
      }
      \label{fig:wavelettrie}
    \end{center}
  \end{minipage}
\end{figure}

We informally define the \emph{Wavelet Trie} of a sequence of binary
strings $S$ as a Wavelet Tree on $S$ (seen as a sequence on the
alphabet $\Sigma = \Sset$) whose tree structure is given by the
Patricia Trie of $\Sset$. We focus on binary strings without loss of
generality, since strings from larger alphabets can be binarized as
described in Section~\ref{sec:preliminaries}. Likewise, we can assume
that $\Sset$ is prefix-free, as any set of strings can be made
prefix-free by appending a terminator symbol to each string.

A formal definition of the Wavelet Trie can be given along the lines
of the Patricia Trie definition of Section~\ref{sec:preliminaries}.

\begin{definition}\label{def:wavelettrie}
  Let $S$ be a non-empty sequence of binary strings, $S \equiv
  \langle s_0,\dots,s_{n-1}\rangle$, $s_i \in \{\bitzero, \bitone\}^*$, whose
  underlying string set $\Sset$ is prefix-free. The \emph{Wavelet
    Trie} of $S$, denoted $\WT(S)$, is built recursively
  as follows:

$(i)$~If the sequence is constant, i.e. $s_i = \alpha$ for all $i$,
    the Wavelet Trie is a node labeled with $\alpha$.

    $(ii)$~Otherwise, let $\alpha$ be the longest common prefix of
    $S$. For any $0 \leq i < n$ we can write $s_i = \alpha b_i
    \gamma_i$, where $b_i$ is a single bit. For $b \in \{\bitzero,
    \bitone\}$ we can then define two sequences $S_b = \langle
    \gamma_i | b_i = b \rangle$, and the bitvector $\beta = \langle
    b_i\rangle$; in other words, $S$ is partitioned in the two
    subsequences depending on whether the string begins with
    $\alpha\bitzero$ or $\alpha\bitone$, the remaining suffixes form
    the two sequences $S_\bitzero$ and $S_\bitone$, and the bitvector
    $\beta$ discriminates whether the suffix $\gamma_i$ is in
    $S_\bitzero$ or $S_\bitone$. Then the Wavelet Trie of $S$ is the
    tree whose root is labeled with $\alpha$ and $\beta$, and whose
    children (respectively labeled with $\bitzero$ and $\bitone$) are
    the Wavelet Tries of the sequences $S_\bitzero$ and $S_\bitone$.
\end{definition}

An example is shown in Fig.~\ref{fig:wavelettrie}. Note that leaves
are labeled only with the common prefix $\alpha$ while internal nodes
are labeled both with $\alpha$ and the bitvector $\beta$.
The Wavelet Trie is a generalization of the Wavelet Tree on $S$: each
node splits the underlying string set $\Sset$ in two subsets and a
bitvector is used to tell which elements of the sequence belong to
which subset. 
Using the same algorithms in \cite{DBLP:conf/soda/GrossiGV03} we obtain the following.
\begin{lemma}
  The Wavelet Trie supports $\Access$, $\Rank$, and $\Select$ operations. In
  particular, if $h_s$ is the number of internal nodes in the
  root-to-node path representing $s$ in $\WT(S)$,
  $\Access(\pos)$ performs $O(h_s)$ $\Rank$ operations on
    the bitvectors, where $s$ is the resulting string;
  $\Rank(s, \pos)$ performs $O(h_s)$ $\Rank$ operations on
    the bitvectors;
  $\Select(s, \occidx)$ performs $O(h_s)$ $\Select$ operations on
    the bitvectors.
\end{lemma}

It is interesting to note that any Wavelet Tree can be seen as a
Wavelet Trie through a specific mapping of the alphabet to binary
strings. For example the classic balanced Wavelet Tree can be obtained
by mapping each element of the alphabet to a distinct string of
$\lceil \log \sigma \rceil$ bits; another popular variant is the
Huffman-tree shaped Wavelet Tree, which can be obtained as a Wavelet
Trie by mapping each symbol to its Huffman code.

\smallskip
\ttlpar{Prefix operations.} It follows immediately from
Definition~\ref{def:wavelettrie} that for any prefix $p$ occurring in
at least one element of the sequence, the subsequence of strings
starting with $p$ is represented by a subtree of $\WT(S)$.

This simple property allows us to support two new operations,
$\RankPrefix$ and $\SelectPrefix$, as defined in the introduction.
The implementation is identical to $\Rank$ and $\Select$, with the
following modifications: if $n_p$ is the node obtained by
prefix-searching $p$ in the trie, for $\RankPrefix$ the top-down
traversal stops at $n_p$; for $\SelectPrefix$ the bottom-up traversal
starts at $n_p$. This proves the following lemma.

\begin{lemma}
  Let $p$ be a prefix occurring in the sequence $S$. Then
  $\RankPrefix(p, \pos)$ performs $O(h_p)$ $\Rank$ operations on
    the bitvectors, and
    $\SelectPrefix(p, \occidx)$ performs $O(h_p)$ $\Select$ operations on
    the bitvectors.
\end{lemma}

Note that, since $\Sset$ is prefix-free, $\Rank$ and $\Select$ on any
string in $\Sset$ are equivalent to $\RankPrefix$ and $\SelectPrefix$,
hence it is sufficient to implement these two operations.

\ttlpar{Average height.} To analyze the space occupied by the Wavelet
Trie, we define the \emph{average height}.

\begin{definition}\label{def:avgheight}
  The average height $\tilde h$ of a $\WT(S)$ is defined as $\tilde h = \frac{1}{n} \sum_{i=0}^{n-1} h_{s_i}$.
\end{definition}

Note that the average is taken on the \emph{sequence}, not on the set
of distinct values. Hence we have $\tilde h n \leq \sum_{i=0}^{n-1}
|s_i|$ (i.e. the total input size), but we expect $\tilde h n \ll
\sum_{i=0}^{n-1} |s_i|$ in real situations, for example if short
strings are more frequent than long strings, or they have long
prefixes in common (exploiting the path compression of the Patricia
Trie). The quantity $\tilde h n$ is equal to the sum of the lengths of all the
bitvectors $\beta$, since each string $s_i$ contributes exactly one
bit to all the internal nodes in its root-to-leaf path. Also, the
root-to-leaf paths form a prefix-free encoding for $\Sset$, and their
concatenation for each element of $S$ is an order-zero encoding for
$S$, thus it cannot be smaller than the zero-order entropy of $S$, as
summarized in the following lemma.

\begin{lemma}\label{lem:hbound}
  Let $\tilde h$ be the average height of $WT(S)$. Then $H_0(S) \leq
  \tilde h \leq \frac{1}{n} \sum_{i=0}^{n-1} |s_i|$.
\end{lemma}

\ttlpar{Static succinct representation.} Our first representation of
the Wavelet Trie is static. We show how by using suitable succinct data
structures the space can be made very close to the information theoretic
lower bound.

To store the Wavelet Trie we need to store its two components: the
underlying Patricia Trie and the bitvectors in the internal nodes. 

We represent the trie using a DFUDS \cite{dfuds} encoding, which
encodes a tree with $k$ nodes in $2k + o(k)$ bits, while supporting
navigational operations in constant time. Since the internal nodes in
the tree underlying the Patricia Trie have exactly two children, we
compute the corresponding tree in the first-child/next-sibling
representation. This brings down the number of nodes from $2|\Sset| -
1$ to $|\Sset|$, while preserving the operations. Hence we can encode
the tree structure in $2|\Sset| + o(|\Sset|)$ bits. If we denote the
number of trie edges as $e = 2(|\Sset| - 1)$, the space can be written
as $e + o(|\Sset|)$.

The $e$ labels $\alpha$ of the nodes are concatenated in depth-first
order in a single bitvector $L$. We use the partial sum data structure
of \cite{RRR07} to delimit the labels in $L$. This adds $\logbinom(e,
|L| + e) + o(|\Sset|)$ bits. The total space (in bits) occupied by the
trie structure is hence $|L| + e + \logbinom(e, |L| + e) +
o(|\Sset|)$.

We now recast the lower bound in \cite{pods08} using our notation,
specializing it for the case of binary strings.
\begin{theorem}[\cite{pods08}]
  \label{thm:ltbound}
  For a prefix-free string set $\Sset$, the information-theoretic
  lower bound $\LT(\Sset)$ for encoding $\Sset$ is given by $\LT(\Sset) =
  |L| + e + \logbinom(e, |L| + e)$, where $L$ is the bitvector
  containing the $e$ labels $\alpha$ of the nodes concatenated in
  depth-first order.
\end{theorem}

It follows immediately that the trie space is just the lower
bound $\LT$ plus a negligible overhead.

It remains to encode the bitvectors $\beta$. We use the RRR encoding,
which takes $|\beta|H_0(\beta) + o(|\beta|)$ to compress the bitvector
$\beta$ and supports constant-time $\Rank$/$\Select$ operations. In
\cite{DBLP:conf/soda/GrossiGV03} it is shown that, regardless of the
shape of the tree, the sum of the entropies of the bitvectors
$\beta$'s add up to the total entropy of the sequence, $nH_0(S)$, plus
negligible terms.

With respect to the redundancy beyond $nH_0(S)$, however, we cannot assume that
$|\Sset| = o(n)$ and that the tree is balanced, as in
\cite{DBLP:conf/soda/GrossiGV03} and most Wavelet Tree literature; in
our applications, it is well possible that $|\Sset| = \Theta(n)$, so a
more careful analysis is needed. In Appendix~\ref{sec:multipleRRR},
Lemma~\ref{lem:redundancywt} we show that in the general case the
redundancy add up to $o(\tilde h n)$ bits.

We concatenate the RRR encodings of the bitvectors, and use again the
partial sum structure of \cite{RRR07} to delimit the encodings, with
an additional space occupancy of $o(\tilde h n)$. The bound is proven
in Appendix~\ref{sec:multipleRRR}, Lemma~\ref{lem:rrrdelim}. Overall, 
the set of bitvectors occupies $nH_0(S) + o(\tilde h n)$ bits.

All the operations can be supported with a trie traversal, which takes
$O(|s|)$ time, and $O(h_s)$ $\Rank$/$\Select$ operations on the
bitvectors. 
Since the bitvector operations are
constant time, all the operations take $O(|s| + h_s)$ time.
Putting together these observations, we obtain the following theorem.

\begin{theorem}\label{thm:wtstatic}
  The Wavelet Trie $WT(S)$ of a sequence of binary strings $S$ can be
  encoded in $\LT(\Sset) + nH_0(S) + o(\tilde h n)$ bits, while
  supporting the operations $\Access$, $\Rank$, $\Select$,
  $\RankPrefix$, and $\SelectPrefix$ on a string $s$ in $O(|s| + h_s)$ time.
\end{theorem}

Note that when the tree is balanced both time and space bounds are
basically equivalent to those of the standard Wavelet Tree. We remark
that the space upper bound in Theorem~\ref{thm:wtstatic} is just the information theoretic
lower bound $\LB(S) \equiv \LT(\Sset) + nH_0(S)$ plus an overhead
negligible in the input size. 

\section{Dynamic Wavelet Tries}\label{sec:dynamic-wavelet-tries}

In this section we show how to implement dynamic updates to the
Wavelet Trie, resulting in the first compressed dynamic sequence with dynamic
alphabet. This is the main contribution of the paper. 

Dynamic variants of Wavelet Trees have been presented recently
\cite{DBLP:journals/jcse/LeeP09, navarrodynamicbitvectors09,
  DBLP:journals/talg/MakinenN08}. They all assume that the alphabet is
known a priori, hence the tree structure is static. Under this
assumption it is sufficient to replace the bitvectors in the nodes
with \emph{dynamic bitvectors with indels}, bitvectors that support
the insertion of deletion of bits at arbitrary points. Insertion at
position $\pos$ can be performed by inserting $\bitzero$ or $\bitone$
at position $\pos$ of the root, whether the leaf corresponding to the
value to be inserted is on the left or right subtree. A $\Rank$
operation is used to find the new position $\pos'$ in the
corresponding child. The algorithm proceeds recursively until a leaf
is reached. Deletion is symmetric.

The same operations can be implemented on a Wavelet Trie. The novelty
consists in the ability of inserting strings that do not already occur in the
sequence, and of deleting the last occurrence of a string, in both
cases changing the alphabet $\Sset$ and thus the shape of the tree.
To do so we represent the underlying tree structure of the Wavelet
Trie with a dynamic Patricia Trie. We summarize the properties of a
dynamic Patricia Trie in the following lemma. The operations are
standard, but we describe them in
Appendix~\ref{sec:patriciaops} for completeness.

\begin{lemma}
  A dynamic Patricia Trie on $k$ binary strings occupies $O(kw) + |L|$
  bits, where $L$ is defined as in Theorem~\ref{thm:ltbound}. Besides
  the standard traversal operations in constant time, 
  insertion of a new string $s$ takes $O(|s|)$ time.
  Deletion of a string $s$ takes $O(\hat \ell)$ time, where $\hat
    \ell$ is the length of the longest string in the trie. %
\end{lemma}

\ttlpar{Updating the bitvectors.} 
Each internal node of the trie is augmented with a bitvector $\beta$,
as in the static Wavelet Trie. Inserting and deleting a
string induce the following changes on the bitvectors $\beta$s.

\smallskip 

$\Insert(s, \pos)$: If the string is not present, we insert it into
the Patricia Trie, causing the split of an existing node: a new
internal node and a new leaf are added. We initialize the bitvector in
the new internal node as a constant sequence of bits $b$ if the
split node is a $b$-labeled child of the new node; 
the length of the new bitvector is equal to the length of
the sequence represented by the split node (i.e.\mbox{} the number of
$b$ bits in the parent node if the split node is a $b$-labeled child). 
The algorithm then
follows as if the string was in the trie.  This operation is shown in
Figure~\ref{fig:wtsplit}.
Now we can assume the string is in the trie. Let prefix $\alpha$ and
bitvector $\beta$ be the labels in the root. Since the string is in
the trie, it must be in the form $\alpha b \gamma$, where $b$ is a
bit. We insert $b$ at position $\pos$ in $\beta$ and compute $\pos' =
\Rank(b, \pos)$ in $\beta$, and insert recursively $\gamma$ in the $b$-labeled
subtree of the root at position $\pos'$. We proceed until
we reach a leaf.

\smallskip

$\Delete(\pos)$: Let $\beta$ be the bitvector in the root. We first
find the bit corresponding to position $\pos$ in the bitvector, $b =
\Access(\pos)$ in $\beta$. Then we compute $\pos' = \Rank(b, \pos)$ in
$\beta$, and delete recursively the string at position $\pos'$ from
the $b$-labeled subtree. We then delete the bit at position $\pos$
from $\beta$.
We then check if the parent of the leaf node representing the string
has a constant bitvector; in this case the string deleted was the last
occurrence in the sequence. We can then delete the string from the
Patricia Trie, thus deleting an internal node (whose bitvector is now constant) and a leaf.

\begin{figure}[htbp]
  \centering
\begin{tikzpicture}[
  scale=0.8, transform shape,
  node distance=3cm,
  intnode/.style={rectangle, draw=black, fill=none, thick, anchor=north},
  leaf/.style={rectangle, draw=black, fill=none, color=black, minimum width=1.5cm, minimum height=0.7cm, thick, anchor=north},
  child anchor=north,
  edge from parent/.style={draw=black, thick},
  subtrie/.style={isosceles triangle, draw=black, shape border rotate=90,isosceles triangle stretches=true, minimum height=15mm,minimum width=12mm,inner sep=0, anchor=apex, thick, xshift=-12mm, yshift=7.5mm},
  level distance=1cm,
  level/.style={sibling distance=2.5cm}
]
  \node [intnode] (s1) {
    $\begin{array}{c}
      \alpha: \bitnobf{\dots}\\
      \beta: \bitnobf{10110\dots}\\
    \end{array}$
  }
  child {
    node [intnode] {
      $\begin{array}{c}
        \alpha: \gamma \bitnobf{0} \delta \\
        \beta: \bitnobf{0111\dots}\\
      \end{array}$
    }
  } 
  child {
    edge from parent [draw=none] {}
  }
  ;

  \node [intnode] (s2) [right=of s1] {
    $\begin{array}{c}
      \alpha: \bitnobf{\dots}\\
      \beta: \bitnobf{10110\dots}\\
    \end{array}$
  }
  child {
    node [intnode] {
      $\begin{array}{c}
        \alpha: {\bm \gamma} \\
        \beta: \bit{0000\dots}\\
      \end{array}$
    }
    child {
      node [intnode] {
        $\begin{array}{c}
          \alpha: {\bm \delta} \\
          \beta: \bitnobf{0111\dots}\\
        \end{array}$
      }
    } 
    child {
      node [leaf] {
        $\begin{array}{c}
          \alpha: {\bm \lambda} \\
        \end{array}$
      }
    }
  } 
  child {
    edge from parent [draw=none] {}
  }
  ;

  \node [intnode] (s3) [right=of s2] {
    $\begin{array}{c}
      \alpha: \bitnobf{\dots}\\
      \beta: \bitnobf{101}\bit{0}\bitnobf{10\dots}\\
    \end{array}$
  }
  child {
    node [intnode] {
      $\begin{array}{c}
        \alpha: \gamma \\
        \beta: \bitnobf{0}\bit{1}\bitnobf{000\dots}\\
      \end{array}$
    }
    child {
      node [intnode] {
        $\begin{array}{c}
          \alpha: \delta \\
          \beta: \bitnobf{0111\dots}\\
        \end{array}$
      }
    } 
    child {
      node [leaf] {
        $\begin{array}{c}
          \alpha: \lambda \\
        \end{array}$
      }
    }
  } 
  child {
    edge from parent [draw=none] {}
  }
  ;
\end{tikzpicture}
  \caption{Insertion of the new string $s = \dots \gamma \bitnobf{1}
    \lambda$ at position $3$. An existing node is split by adding a
    new internal node with a constant bitvector and a new leaf. The
    corresponding bits are then inserted in the root-to-leaf path
    nodes.}
  \label{fig:wtsplit}
\end{figure}
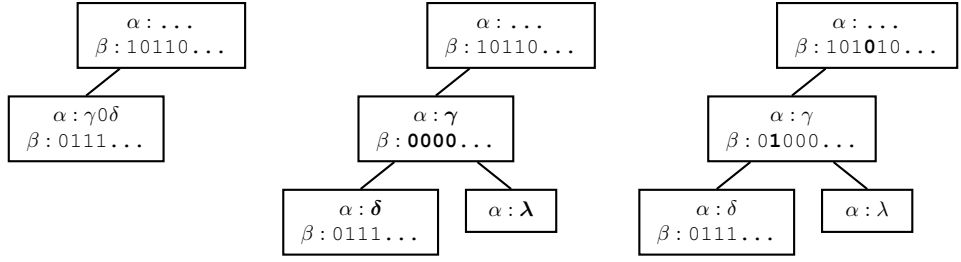

In both cases the number of operations ($\Rank$,
$\Insert$, $\Delete$) on the bitvectors is bounded by $O(h_s)$.
The operations we need to perform on the bitvectors are the standard
insertion/deletion, with one important exception: when a node is
split, we need to create a new constant bitvector of arbitrary
length. We call this operation $\Init(b, n)$, which fills an empty
bitvector with $n$ copies of the bit $b$. The following remark
rules out for our purposes most existing dynamic bitvector constructions.

\begin{observation}\label{obs:init}
  If the encoding of a constant (i.e. $\bitzero^n$ or $\bitone^n$)
  bitvector uses $\omega(f(n))$ memory words (of size $w$), $\Init(b, n)$ cannot be
  supported in $O(f(n))$ time.
\end{observation}

Uncompressed bitvectors use $\Omega(n / w)$ words; the
compressed bitvectors of \cite{DBLP:journals/talg/MakinenN08,
  navarrodynamicbitvectors09}, although they have a desirable occupancy
of $|\beta|H_0(\beta) + o(|\beta|)$, have $\Omega(n \log \log
    n / (w \log n))$ words of redundancy. Since we aim for polylog
operations, these constructions cannot be considered \emph{as is}.

\smallskip
\ttlpar{Main results.} 
We first consider the case of append-only sequences. We remark that, in the $\Insert$
operation described above,
when appending a string at the end of the sequence the bits inserted in the bitvectors are
appended at the end, so it is sufficient that the bitvectors support an $\Append$ operation in
place of a general $\Insert$. Furthermore, $\Init$ can
be implemented simply by adding a left offset in each bitvector, which
increments each bitvector space by $O(\log n)$ and can be checked in
constant time. Using the append-only bitvectors described in
Section~\ref{sub:append-only-bitvectors}, and observing that the
redundancy is as in Section~\ref{sec:wavelet-trie}, we can state the
following theorem.

\begin{theorem}
  \label{the:append-wt}
  The append-only Wavelet Trie on a dynamic sequence $S$ supports the
  operations $\Access$, $\Rank$, $\Select$, $\RankPrefix$,
  $\SelectPrefix$, and $\Append$ in $O(|s| + h_s)$ time.  The total
  space occupancy is $O(|\Sset|w) + |L| + nH_0(S) + o(\tilde{h} n)$ bits, where $L$
  is defined as in Theorem~\ref{thm:ltbound}.
\end{theorem}

Using instead the fully-dynamic bitvectors in
Section~\ref{sub:fully-dynaminc-bitvectors}, we can state the following theorem.

\begin{theorem}
  \label{the:dynamic-wt}
  The dynamic Wavelet Trie on a dynamic sequence $S$ supports the
  operations $\Access$, $\Rank$, $\Select$, $\RankPrefix$,
  $\SelectPrefix$, and $\Insert$ in $O(|s| + h_s \log n)$
  time. $\Delete$ is supported in $O(|s| + h_s \log n)$ time if $s$
  occurs more than once, otherwise time is $O(\hat \ell + h_s \log
  n)$, where $\hat \ell$ is the length of the longest string.
  The total space occupancy is $O(nH_0(S) + |\Sset|w) + L$ bits,
  where $L$ is defined as in Theorem~\ref{thm:ltbound}.
\end{theorem}

Note that, using the compact notation defined in the introduction, the
space bound in Theorem~\ref{the:append-wt} can be written as $\LB(S) +
\PT(\Sset) + o(\tilde h n)$, while the one in
Theorem~\ref{the:dynamic-wt} can be written as $\LB(S) + \PT(\Sset) + O(nH_0)$.

\subsection{Append-only bitvectors}
\label{sub:append-only-bitvectors}

In this section we describe an append-only bitvector with
constant-time $\Rank$/$\Select$/$\Append$ operations and
nearly-optimal space occupancy. The data structure uses RRR as a
black-box data structure, assuming only its query time and space
guarantees. We require the following \emph{decomposable} property on RRR:
given an input bitvector of $n$ bits packed into $O(n/w)$ words of
size $w \geq \log n$, RRR can be built in $O(n' / \log n)$ time for
any chunk of $n' \geq \log n$ consecutive bits of the input bitvector,
using table lookups and the Four-Russians trick; moreover, this $O(n'
/ \log n)$-time work can be spread over $O(n' / \log n)$ steps, each
of $O(1)$ time, that can be interleaved with other operations not
involving the chunk at hand. This a quite mild requirement and, for
this reason, it is a general technique that can be applied to other
static compressed bitvectors other than RRR with the same
guarantees. Hence we believe that the following approach is
of independent interest.

\begin{theorem}
  \label{the:append-only-bitvectors}
  The append-only bitvector supports $\Access$, $\Rank$, $\Select$,
  and $\Append$ on a bitvector $\beta$ in $O(1)$ time. The total space
  is $nH_0(\beta) + o(n)$ bits, where $n = |\beta|$.
\end{theorem}

Before describing the data structure and proving
Theorem~\ref{the:append-only-bitvectors} we need to introduce some
auxiliary lemmas.

\begin{lemma}[Small Bitvectors]
  \label{lem:small-bitvectors}
  Let $\beta'$ be a bitvector of bounded size $n' =
  O(\polylog(n))$. Then there is a data structure that supports
  $\Access$, $\Rank$, $\Select$, and $\Append$ on $\beta'$ in $O(1)$
  time, while occupying $O(\polylog(n))$ bits.
\end{lemma}
\begin{proof}
  It is sufficient to store explicitly all the answers to the queries
  $\Rank$ and $\Select$ in arrays of $n'$ elements, thus taking
  $O(n'\log n') = O(\polylog(n))$. $\Append$ can be supported in
  constant time by keeping a running count of the $\bitone$s in the
  bitvector and the position of the last $\bitzero$ and $\bitone$,
  which are sufficient to compute the answers to the $\Rank$ and
  $\Select$ queries for the appended bit.
\end{proof}

\begin{lemma}[Amortized constant-time]
  \label{lem:amortized-constant-append}
  There is a data structure that supports $\Access$, $\Rank$, and
  $\Select$ in $O(1)$ time and $\Append$ in \emph{amortized} $O(1)$
  time on a bitvector $\beta$ of $n$ bits. The total space occupancy
  is $nH_0(\beta) + o(n)$ bits.
\end{lemma}
\begin{proof}
  We split the input bitvector $\beta$ into $t$ smaller bitvectors
  $V_t, V_{t-1}, \ldots, V_1$, such that $\beta$ is equal to the
  concatenation $V_t \cdot V_{t-1} \cdots V_1$ at any time. Let $n_i =
  |V_i| \geq 0$ be the length of $V_i$, and $m_i$ be the number of
  $\bitone$s in it, so that $\sum_{i=1}^t m_i = m$ and $\sum_{i=1}^t
  n_i = n$.  Following Overmars's logarithmic
  method~\cite{Overmars:1983:DDD}, we maintain a collection of static
  data structures on $V_t, V_{t-1}, \ldots, V_1$ that are periodically
  rebuilt.
  \begin{enumerate}[(a)]
    \shrinkitems
  \item \label{item:SD} A data structure $F_1$ as described in
    Lemma~\ref{lem:small-bitvectors} to store $\beta' = V_1$. Space is
    $O(\polylog(n))$ bits.
  \item \label{item:FD} A collection of static data structures $F_t,
    F_{t-1}, \ldots, F_2$, where each $F_i$ stores $V_i$ using
    RRR. Space occupancy is $nH_0(\beta) + o(n)$ bits.
  \item \label{item:PD} Fusion Trees~\cite{JCSS::FredmanW1993} of
    constant height storing the partial sums on the number of
    $\bitone$s, $s_i^\bitone = \sum_{j=t}^{i+1} m_j$, where
    $s_t^\bitone=0$, and symmetrically the partial sums on the number
    of $\bitzero$s, $s_i^\bitzero = \sum_{j=t}^{i+1} (n_j-m_j)$,
    setting $s_t^\bitone=0$.  $\Pred$ takes $O(1)$ time and
    construction is $O(t)$ time. Space occupancy is $O(t \log n) =
    o(n)$ bits.
  \end{enumerate}

  We fix $r = c \log n_0$ for a suitable constant $c>1$, where $n_0$
  is the length $n > 2$ of the initial input bitvector $\beta$. We
  keep this choice of $r$ until $F_t$ is reconstructed: at that point,
  we set $n_0$ to the current length of $\beta$ and we update $r$
  consistently. Based on this choice of $r$, we guarantee that $r =
  \Theta(\log n)$ at any time and introduce the following constraints:
  $n_1 \leq r$ and, for every $i > 1$, $n_i$ is either $0$ or $2^{i -
    2} r$. It follows immediately that $t = \Theta(\log n)$, and hence
  the Fusion Trees in~(\ref{item:PD}) contain $O(\log n)$ entries,
  thus guaranteeing constant height.

  We now discuss the query operations. $\Rank(b, \pos)$ and
  $\Select(b, \occidx)$ are performed as follows for a bit $b \in \{
  \bitzero, \bitone \}$. Using the data structure in~(\ref{item:PD}),
  we identify the corresponding bitvector $V_i$ along with the number
  $s_i^b$ of occurrences of bit~$b$ in the preceding ones, $V_t,
  \ldots, V_{i+1}$. The returned value corresponds to the index $i$ of
  $F_i$, which we query and combine the result with $s_i^b$: we output
  the sum of $s_i^b$ with the result of $\Rank(b, \pos -
  \sum_{j=t}^{i+1} n_i)$ query on $F_i$ in the former case; we output
  $\Select(b, \occidx - s_i^b)$ query on $F_i$ in the latter. Hence,
  the cost is $O(1)$ time.

  It remains to show how to perform $\Append(b)$ operation. While $n_1
  < r$ we just append the bit $b$ to $F_1$, which takes constant time
  by Lemma~\ref{lem:small-bitvectors}. When $n_1$ reaches $r$, let $j$
  be the smallest index such that $n_j = 0$. Then $\sum_{i=1}^{j-1}
  n_i = 2^{j - 2} r$, so we concatenate $V_{j-1}\cdots V_1$ and rename
  this concatenation $V_j$ (no collision since it was $n_j = 0$). We
  then rebuild $F_j$ on $V_j$ and set $F_i$ for $i < j$ to empty
  (updating $n_j, \ldots, n_1$). We also rebuild the Fusion Trees
  of~(\ref{item:PD}), which takes an additional $O(\log n)$ time.
  When $F_t$ is rebuilt, we have that the new $V_t$ corresponds to the
  whole current bitvector $\beta$, since $V_{t-1}, \ldots, V_1$ are
  empty. We thus set $n_0 := |\beta|$ and update $r$ consequently. By
  observing that each $F_j$ is rebuilt every $O(n_j)$ $\Append$
  operations and that RRR construction time is $O(n_j/\log n)$, it
  follows that each $\Append$ is charged $O(1/\log n)$ time on each
  $F_j$, thus totaling $O(t/\log n) = O(1)$ time.
\end{proof}

We now show how to de-amortize the data structure in
Lemma~\ref{lem:amortized-constant-append}. In the de-amortization we
have to keep copies of some bitvectors, so the $nH_0$ term becomes
$O(nH_0)$.

\begin{lemma}[Redundancy]
  \label{lem:deamortized-constant-append}
  There is a data structure that supports $\Access$, $\Rank$,
  $\Select$, and $\Append$ in $O(1)$ time on a bitvector $\beta$ of
  $n$ bits. The total space occupancy is $O(nH_0(\beta)) + o(n)$.
\end{lemma}
\begin{proof}
  To de-amortize the structure we follow Overmars's classical method
  of partial rebuilding \cite{Overmars:1983:DDD}. The idea is to
  spread the construction of the RRR's $F_j$ over the next $O(n_j)$
  $\Append$ operations, charging extra $O(1)$ time each. We already
  saw in Lemma~\ref{lem:amortized-constant-append} that this suffices
  to cover all the costs. Moreover, we need to increase the speed of
  construction of $F_j$ by a suitable constant factor with respect to
  the speed of arrival of the $\Append$ operations, so we are
  guaranteed that the construction of $F_j$ is completed before the
  next construction of $F_j$ is required by the argument shown in the
  proof of Lemma~\ref{lem:amortized-constant-append}. We refer the
  reader to \cite{Overmars:1983:DDD} for a thorough discussion of the
  technical details of this general technique.

  While $V_1$ reaches its bound of $r$ bits, we have a budget of
  $\Theta(r) = \Theta(\log n)$ operations that we can use to
  prepare the next version of the data structure. We use this budget
  to perform the following operations.
  
  \begin{enumerate}[(1)]
  \item Identify the smallest $j$ such that $n_j = 0$ and start the
    construction of $F_j$ by creating a \emph{proxy bitvector} $\tilde
    F_j$ which references the existing $F_{j-1},\dots,F_1$ and Fusion
    Trees in~(\ref{item:PD}), so that it can answer queries in $O(1)$
    time as if it was the fully built $F_j$. When we switch to this
    version of the data structure, these $F_{j-1},\dots,F_1$ become
    accessible only inside $\tilde F_j$.
    
  \item Build the Fusion Trees in~(\ref{item:PD}) for the next
    reconstruction of the data structure. Note that this would require
    to know the final values of the $n_i$s and $m_i$s when $V_1$ is
    full and the reconstruction starts. Instead, we use the current
    values of $n_i$ and $m_i$: only the values for the last non-empty
    segment will be wrong. We can \emph{correct} the Fusion Trees by
    adding an additional \emph{correction} value to the last non-empty
    segment; applying the correction at query time has constant-time
    overhead.

  \item Build a new version of the data structure which references the
    new Fusion Trees, the existing bitvectors $F_t,\dots,F_{j+1}$, the
    proxy bitvector $\tilde F_j$ and new empty bitvectors
    $F_{j-1},\dots,F_1$ (hence, $n_j = 2^{j-2} r$ and $n_{j-1} =
    \cdots = n_1 = 0$).
  \end{enumerate}

  When $n_1$ reaches $r$, we can replace in constant time the data
  structure with the one that we just finished rebuilding.

  At each $\Append$ operation, we use an additional $O(1)$ budget to
  advance the construction of the $F_j$s from the proxies $\tilde
  F_j$s in a round-robin fashion. When the construction of one $F_j$
  is done, the proxy $\tilde F_j$ is discarded and replaced by
  $F_j$. Since, by the amortization argument in the proof of
  Lemma~\ref{lem:amortized-constant-append}, each $F_j$ is completely
  rebuilt by the time it has to be set to empty (and thus used for the
  next reconstruction), at most one copy of each bitvector has to be
  kept, thus the total space occupancy grows from $nH_0(\beta) + o(n)$
  to $O(nH_0(\beta)) + o(n)$.  Moreover, when $r$ has to increase (and
  thus the $n_i$'s should be updated), we proceed as
  in~\cite{Overmars:1983:DDD}.
\end{proof}

We can now use the de-amortized bitvector to bootstrap a constant-time
append-only bitvector with space occupancy $nH_0(\beta) + o(n)$, thus
proving Theorem~\ref{the:append-only-bitvectors}.
\begin{proof}[of Theorem~\ref{the:append-only-bitvectors}]
  
  Let $\beta$ be the input bitvector, and $L = \Theta(\polylog(n))$ be a
  power of two. 
  We split $\beta$ into $n_L \equiv
  \lfloor n/L \rfloor$ smaller bitvectors $B_i$'s, each of length $L$
  and with $\hat m_i$ $\bitone$s ($0 \leq \hat m_i \le L$), plus a
  residual bitvector $B'$ of length $0 \leq |B'|<L$: at any time $\beta =
  B_1 \cdot B_2 \cdots B_{n_L} \cdot B'$. Using
  this partition, we maintain the following data structures:

  \begin{enumerate}[(1)]
    \shrinkitems
  \item \label{item:BD} A collection $\hat F_1$, $\hat F_2$,~\dots,
    $\hat F_{n_L}$ of static data structures, where each $\hat F_i$
    stores $B_i$ using RRR.
  \item \label{item:FB} The data structure in
    Lemma~\ref{lem:small-bitvectors} to store $B'$.
  \item \label{item:TD} The data structure in
    Lemma~\ref{lem:deamortized-constant-append} to store the partial
    sums $\hat s_i^\bitone = \sum_{j=1}^{i-1} \hat m_j$, setting $\hat
    s_1^\bitone=0$.  This is implemented by maintaining a bitvector
    that has a $\bitone$ for each position $\hat s_i^\bitone$, and
    $\bitzero$ elsewhere. $\Pred$ queries can be implemented by
    composing $\Rank$ and $\Select$. The bitvector has length $n_L+m$
    and contains $n_L$ $\bitone$s. The partial sums $\hat s_i^\bitzero
    = \sum_{j=1}^{i-1} (L - \hat m_j)$ are kept symmetrically in
    another bitvector.
  \end{enumerate}

  $\Rank(b, \pos)$ and $\Select(b, \occidx)$ are implemented as follows
  for a bit $b \in \{ \bitzero, \bitone \}$. Using
  the data structure in~(\ref{item:TD}), we identify the corresponding
  bitvector $B_i$ in~(\ref{item:BD}) or $B'$ in~(\ref{item:FB}) along
  with the number $\hat s_i^b$ of occurrences of bit~$b$ in the preceding segments.
  In both cases, we query the corresponding
  dictionary and combine the result with $\hat s_i^b$. These operations take
  $O(1)$ time.

  Now we focus on $\Append(b)$. At every $\Append$ operation, we
  append a $\bitzero$ to the one of the bitvectors in~(\ref{item:TD})
  depending on whether $b$ is $\bitzero$ or $\bitone$, thus
  maintaining the partial sums invariant. This takes constant time. We
  guarantee that $|B'| \leq L$ bits: whenever $|B'| = L$, we
  conceptually create $B_{n_L+1} := B'$, still keeping its data
  structure in~(\ref{item:FB}); reset $B'$ to be empty, creating the
  corresponding data structure in~(\ref{item:FB}); append a $\bitone$
  to the bitvectors in~(\ref{item:TD}).  We start building a new
  static compressed data structure $\hat F_{n_L+1}$ for $B_{n_L+1}$
  using RRR in $O(L/\log n)$ steps of $O(1)$ time each.
  During the construction of $\hat F_{n_L+1}$ the old $B'$ is still
  valid, so it can be used to answer the queries. As soon as the
  construction is completed, in $O(L / \log n)$ time, the old $B'$ can
  be discarded and queries can be now handled by $\hat F_{n_L+1}$.
  Meanwhile the new appended bits are handled in the new
  $B'$, in $O(1)$ time each, using its new instance
  of~(\ref{item:FB}).  By suitably tuning the speed of the operations,
  we can guarantee that by the time the new reset $B'$ has reached
  $L/2$ (appended) bits, the above $O(L)$ steps have been completed
  for $\hat F_{n_L+1}$. Hence, the total cost of $\Append$ is just
  $O(1)$ time in the worst case.

  To complete our proof, we have to discuss what happens when
  we have to double $L := 2 \times L$. This is a standard task known
  as global rebuilding~\cite{Overmars:1983:DDD}. We rebuild RRR for
  the concatenation of $B_1$ and $B_2$, and deallocate the latter two
  after the construction; we then continue with RRR on the
  concatenation of $B_3$ and $B_4$, and deallocate them after the
  construction, and so on. Meanwhile, we build a copy~(\ref{item:TD}')
  of the data structure in~(\ref{item:TD}) for the new parameter $2
  \times L$, following an incremental approach. At any time, we only
  have~(\ref{item:TD}') and $\hat F_{2i-1}, \hat F_{2i}$
  duplicated. The implementation of $\Rank$ and $\Select$ needs a
  minor modification to deal with the already rebuilt segments.  The
  global rebuilding is completed before we need again to double the
  value of $L$.

  We now perform the space analysis. As for~(\ref{item:BD}), we have
  to add up the space taken by $\hat F_1, \ldots, \hat F_{n_L}$ plus
  that taken by the one being rebuilt using $\hat F_{2i-1}, \hat
  F_{2i}$. This sum can be upper bounded by $\sum_{i=1}^{n_L}
  (\logbinom(m_i,L) + o(L) ) + O(L) = H_0(\beta) + o(n)$.  The space
  for~(\ref{item:FB}) is $O(\polylog(n)) = o(n)$.  Finally, the
  occupancy of the $s_i^\bitone$ partial sums in~(\ref{item:TD}) is
  $\logbinom(n_L, n_L+m) + o(n_L+m) = O(n_L\log (1+m/n_L)) = O(n
  \log n / L) = o(n)$ bits, since the bitvector has length $n_L + m$ and contains $n_L$
  $\bitone$s. The analysis is symmetric for the $s_i^\bitzero$ partial
  sums, and for the copies in~(\ref{item:TD}').
\end{proof}

\subsection{Fully dynamic bitvectors}
\label{sub:fully-dynaminc-bitvectors}

We introduce a
new dynamic bitvector construction which, although the entropy term
has a constant greater than $1$, supports logarithmic-time $\Init$ and
$\Insert$/$\Delete$.

To support both insertion/deletion and initialization in logarithmic
time we adapt the dynamic bitvector presented in Section~3.4 of
\cite{DBLP:journals/talg/MakinenN08}; in the paper, the bitvector is
compressed using Gap Encoding, i.e. the bitvector
$\bitzero^{g_0}\bitone\bitzero^{g_1}\bitone\dots$
is encoded as the sequence of \emph{gaps} $g_0, g_1, \dots$, and
the gaps are encoded using Elias delta code
\cite{elias1975universal}. The resulting bit stream is split in chunks
of $\Theta(\log n)$ (without breaking the codes) and a self-balancing
binary search tree is built on the chunks, with partial counts in all
the nodes. Chunks are split and merged upon insertions and deletions
to maintain the chunk size invariant, and the tree rebalanced.

Because of gap encoding, the space has a linear dependence on
the number of \bitone{}s, hence by Remark~\ref{obs:init} it is
not suitable for our purposes. We make a simple modification that
enables an efficient $\Init$: in place of gap encoding and delta codes
we use RLE and Elias gamma codes \cite{elias1975universal}, as the
authors of \cite{DBLP:journals/talg/FoschiniGGV06} do in their
\emph{practical dictionaries}. RLE encodes the bitvector
$\bitzero^{r_0}\bitone^{r_1}\bitzero^{r_2}\bitone^{r_3}\dots$ with the
sequence of \emph{runs} $r_0, r_1, r_2, r_3, \dots$. The runs are
encoded with Elias gamma codes. 
$\Init(b, n)$ can be trivially supported by creating a tree with a
single leaf node, and encoding a run of $n$ bits $b$ in the node, which
can be done in time $O(\log n)$.
In
\cite{DBLP:journals/iandc/FerraginaGM09} it is proven the space of
this encoding is bounded by $O(n H_0)$, but even if the coefficient of
the entropy term is not $1$ as in RRR bitvectors, the experimental
analysis performed in \cite{DBLP:journals/talg/FoschiniGGV06} shows
that RLE bitvectors perform extremely well in practice.
The rest of
the data structure is left unchanged; we refer to
\cite{DBLP:journals/talg/MakinenN08} for the details.

\begin{theorem}
  The dynamic RLE$+\gamma$ bitvector supports $\Access$, $\Rank$, $\Select$,
  $\Insert$, $\Delete$, and $\Init$ on a bitvector $\beta$ in $O(\log
  n)$ time. The total space occupancy is $O(nH_0(\beta) + \log n)$ bits.
\end{theorem}

\section{Other Query Algorithms}
\label{sec:query-algor-wavel}

In this section we describe range query algorithms on the Wavelet Trie
that can be useful in particular in database applications and
analytics. We note that the algorithms for \emph{distinct values in
  range} and \emph{range majority element} are similar to the
\emph{report} and \emph{range quantile} algorithms presented in
\cite{Gagie2010}; we restate them here for completeness, extending
them to prefix operations.
In the following we denote with $\Cop$ the cost of
$\Access$/$\Rank$/$\Select$ on the bitvectors; $\Cop$ is $O(1)$ for
static and append-only bitvectors, and $O(\log n)$ for fully
dynamic bitvectors.

\smallskip
\ttlpar{Sequential access.}
Suppose we want to enumerate all the strings in the range $[l, r)$,
i.e. $S^{[l, r)}=S_l, \dots, S_{r-1}$. We could do it with $r - l$ calls to
$\Access$, but accessing each string $s_i$ would cost $O(|s_i| +
h_{s_i}\Cop)$. We show instead how to enumerate the values of a
range by enumerating the bits of each bitvector: suppose we have an
iterator on root bitvector for the range $[l, r)$. Then if the current
bit is $0$, the next value is the next value given by the left
subtree, while if it is $1$ the next value is the next value of the
right subtree. We proceed recursively by keeping an iterator on
all the bitvectors of the internal nodes we traverse during the
enumeration. 

When we traverse an internal node for the first time, we perform a
$\Rank$ to find the initial point, and create an iterator. Next time
we traverse it, we just advance the iterator. Note that both RRR
bitvectors and RLE bitvectors can support iterators with $O(1)$
advance to the next bit.

By using iterators instead of performing a $\Rank$ each time we
traverse a node, a single $\Rank$ is needed for each traversed node, 
hence to extract the $i$-th string it takes
$O(|s_i| + \frac{1}{r - l}\sum_{s \in S^{[l,
    r)}_{\text{set}}}h_{s}\Cop)$ amortized time.

Note that if $S^{[l, r)}_{\text{set}}$ (the set of distinct strings occurring
in $S^{[l, r)}$) is large, the actual time is smaller due to shared nodes
in the paths. In fact, in the extreme case when the whole string set
occurs in the range, we can give a better bound, $O(|s_i| + \frac{1}{r
  - l}|\Sset|\Cop)$ amortized time.

\smallskip
\ttlpar{Distinct values in range.}
Another useful query is the enumeration of the distinct values in the
range $[l, r)$, which we called $S^{[l, r)}_{\text{set}}$. Note that
for each node the distinct values of the subsequence represented by
the node are just the distinct values of the left subtree plus the
distinct values of the right subtree in the corresponding
ranges. Hence, starting at the root, we compute the number of
$\bitzero$s in the range $[l, r)$ with two calls to $\Rank$. If there
are no $\bitzero$s we just enumerate the distinct elements of the
right child in the range $[\Rank(\bitone, l), \Rank(\bitone, r))$. If
there are no $\bitone$s, we proceed symmetrically. If there are both
$\bitzero$s and $\bitone$s, the distinct values are the union of the
distinct values of the left child in the range $[\Rank(\bitzero, l),
\Rank(\bitzero, r))$ and those of the right child in the range
$[\Rank(\bitone, l), \Rank(\bitone, r))$.
Since we only traverse nodes that lead to values that are in the
range, the total running time is $O(\sum_{s \in S^{[l,
    r)}_{\text{set}}}|s| + h_{s}\Cop)$, which is the same time as
accessing the values, if we knew their positions. As a byproduct, we
also get the number of occurrences of each value in the range.

We can stop early in the traversal, hence enumerating the
distinct \emph{prefixes} that satisfy some property. For example in an
URL access log we can find efficiently the distinct hostnames in a
given time range.

\smallskip
\ttlpar{Range majority element.}
The previous algorithm can be modified to check if there is a majority
element in the range (i.e. one element that occurs more than $\frac{r
  - l}{2}$ times in $[l, r)$), and, if there is such an element, find it. Start at
the root, and count the number of $\bitzero$s and $\bitone$s in the
range. If a bit $b$ occurs more than $\frac{r - l}{2}$ times (note
that there can be at most one) proceed recursively on the $b$-labeled
subtree, otherwise there is no majority element in the range.

The total running time is $O(h \Cop)$, where $h$ is the height of
the Wavelet Trie. In case of success, if the string found is $s$, the
running time is just $O(h_s \Cop)$.
As for the distinct values, this can be applied to prefixes as well by
stopping the traversal when the prefix we found until that point
satisfies some property.

A similar algorithm can be used as an heuristic to find all the values
that occur in the range at least $t$ times, by proceeding as in the
enumeration of distinct elements but discarding the branches whose bit
has less than $t$ occurrences in the parent. While no theoretical
guarantees can be given, this heuristic should perform very well with
power-law distributions and high values of $t$, which are the cases of
interest in most data analytics applications.

\section{Probabilistically-Balanced Dynamic Wavelet Trees}\label{sec:probabilistic-balanced-dynamic-wt}

In this section we show how to use the Wavelet Trie to
maintain a dynamic Wavelet Tree on a sequence from a bounded alphabet
with operations that with high probability do not depend on the universe size.

A compelling example is given by numeric data: a sequence of 
integers, say in $\{0, \dots, 2^{64} - 1\}$, cannot be represented with
existing dynamic Wavelet Trees unless the tree is built on the whole
universe, even if the sequence only contains integers from a much
smaller subset. Similarly, a text sequence in Unicode typically contains
few hundreds of distinct characters, far fewer than the
$\approx 2^{17}$ (and growing) defined in the standard. 

Formally, we wish to maintain a sequence of symbols $S = \langle s_0, \dots, s_{n-1}\rangle$
drawn from an alphabet $\Sigma \subseteq U = \{0, \dots, u - 1\}$, where
we call $U$ the \emph{universe} and $\Sigma$ the \emph{working
  alphabet}, with $\Sigma$ typically much smaller than $U$ and not
known a priori. We want to support the standard $\Access$, $\Rank$,
$\Select$, $\Insert$, and $\Delete$ but we are willing to give up
$\RankPrefix$ and $\SelectPrefix$, which would not make sense anyway
for non-string data.

We can immediately use the Wavelet Trie on $S$, by mapping injectively
the values of $U$ to strings of length $\lceil \log u \rceil$. This supports all the
required operations with a space bound that depends only
logarithmically in $u$, but the height of the resulting trie could be
as much as $\log u$, while a balanced tree would require only $\log
|\Sigma|$.

To obtain a balanced tree without having to deal with complex
rotations we employ a simple randomized technique that will yield a
balanced tree with high probability.
The main
idea is to randomly permute the universe $U$ with an easy to compute
permutation, such that the probability that the alphabet $\Sigma$ will
produce an unbalanced trie is negligibly small.

To do so we use the hashing technique described in
\cite{universalhash97}. We map the universe $U$ onto itself by the
function $h_a(x) = a x \pmod{2^{\lceil \log u \rceil}}$ where $a$ is
chosen at random among the odd integers in $[1, 2^{\lceil \log u
  \rceil} - 1]$ when the data structure is initialized. Note that
$h_a$ is a bijection, with the inverse given by $h^{-1}(x) = a^{-1} x
\pmod{2^{\lceil \log u \rceil}}$. The result of the hash function is
considered as a binary string of $\lceil \log u \rceil$ bits written
LSB-to-MSB, and operations on the Wavelet Tree are defined by
composition of the hash function with operations on the Wavelet Trie;
in other words, the values are hashed and inserted in a Wavelet Trie,
and when retrieved the hash inverse is applied.

To prove that the resulting trie is balanced we use the following
lemma from \cite{universalhash97}.
\begin{lemma}[\cite{universalhash97}]
  Let $\Sigma \subseteq U$ be any subset of the universe, and $\ell =
  \lceil (\alpha + 2)\log |\Sigma| \rceil$ so that $\ell \leq \lceil \log u
  \rceil$. Then the following holds
  \begin{equation*}
    \Prob\big(  \forall x, y \in \Sigma \quad h_a(x) \not\equiv h_a(y) \pmod{2^\ell} \big) \geq 1 - |\Sigma|^{-\alpha}
  \end{equation*}
  where the probability is on the choice of $a$.
\end{lemma}

In our case, the lemma implies that with very high probability the hashes
of the values in $\Sigma$ are distinguished by the first $\ell$ bits,
where $\ell$ is logarithmic in the $|\Sigma|$. The trie on the
hashes cannot be higher than $\ell$, hence it is balanced.
The space occupancy is that of the Wavelet Trie built on the
hashes. We can bound $L$, the sum of trie labels in
Theorem~\ref{thm:ltbound}, by the total sum of the hashes length,
hence proving the following theorem.

\begin{theorem}
  The randomized Wavelet Tree on a dynamic sequence $S = \langle s_0, \dots,
  s_{n-1}\rangle$ where $s_i \in \Sigma \subseteq U = \{0, \dots, u - 1\}$
  supports the operations $\Access$, $\Rank$, $\Select$, $\Insert$,
  and $\Delete$ in time $O(\log u + h \log n)$, where $h \leq (\alpha + 2)\log
  |\Sigma|$ with probability $1 - |\Sigma|^{-\alpha}$ (and $h \leq
  \lceil \log u \rceil$ in the worst case).

  The total space occupancy is $O(nH_0(S) + |\Sigma| w) + |\Sigma| \log u$ bits.
\end{theorem}

\section{Conclusion and Future Work}

We have presented the Wavelet Trie, a new data structure for
maintaining compressed sequences of strings with provable time and
compression bounds. We believe that the Wavelet Trie will find
application in real-world storage problems where space-efficiency is
crucial. To this end, we plan to evaluate the practicality of the data
structure with an experimental analysis on real-world data,
evaluating several performance/space/functionality trade-offs. We are
confident that a properly engineered implementation can perform well,
as other algorithm engineered succinct data structures have proven
very practical (\cite{sadaalx07, bpalx10, compressedtries12}).

It would be also interesting to \emph{balance} the Wavelet Trie, even
for pathological sets of strings. In \cite{compressedtries12} it was
shown that in practice the cost of unbalancedness can be high. 
Lastly, it is an open question how the Wavelet Trie would perform in
external or cache-oblivious models. A starting point would be a
fanout larger than $2$ in the trie, but internal nodes would require
vectors with non-binary alphabet. The existing non-binary dynamic
sequences do not directly support $\Init$, hence they are not suitable
for the Wavelet Trie.

\section*{Acknowledgments}
We would like to thank Ralf Herbrich for suggesting the problem of
compressing relations, and Rossano Venturini for numerous insightful
discussions. We would also like to thank the anonymous reviewers for
several helpful comments and suggestions which helped improving the
presentation.

\bibliographystyle{abbrv}
\bibliography{wavelet_trie_paper}

\newpage
\appendix
\bigskip

\section*{APPENDIX}

\section{Multiple static bitvectors}\label{sec:multipleRRR}

\begin{lemma}\label{lem:entropymin}
  Let $S$ be a sequence of length $n$ on an alphabet of cardinality
  $\sigma$, with each symbol of the alphabet occurring at least
  once. Then the following holds:
  \begin{equation*}
    n H_0(S) \geq (\sigma - 1) \log n \text{.}
  \end{equation*}
\end{lemma}
\begin{proof}
  The inequality is trivial when $\sigma = 1$. When there are at least
  two symbols, the minimum entropy is attained when $\sigma - 1$
  symbols occur once and one symbol occurs the remaining $n - \sigma +
  1$ times. To show this, suppose by contradiction that the minimum
  entropy is attained by a string where two symbols occur more than
  once, occurring respectively $a$ and $b$ times. Their contribution
  to the entropy term is $a \log \frac{n}{a} + b \log
  \frac{n}{b}$. This contribution can be written as $f(a)$ where 
  \[
  f(t) = t \log \frac{n}{t} + (b + a - t) \log \frac {n}{b + a - t},
  \]
  but
  $f(t)$ has two strict minima in $1$ and $b + a - 1$ among the
  positive integers, so the entropy term can be lowered by making one
  of the symbol absorb all but one the occurrences of the other,
  yielding a contradiction.

  To prove the lemma, it is sufficient to see that the contribution to
  the entropy term of the $\sigma - 1$ singleton symbols is $(\sigma -
  1) \log n$.
\end{proof}

\begin{lemma}\label{lem:sigmanegligible}
  $O(|\Sset|)$ is bounded by $o(\tilde h n)$.
\end{lemma}
\begin{proof}
  It suffices to prove that
  \begin{equation*}
    \frac{|\Sset|}{\tilde h n}
  \end{equation*}
  is asymptotic to $0$ as $n$ grows. By Lemma~\ref{lem:hbound} and
  Lemma~\ref{lem:entropymin}, and assuming $|\Sset| \geq 2$,
  \begin{equation*}
    \frac{|\Sset|}{\tilde h n} \leq \frac{|\Sset|}{n H_0(S)} \leq \frac{|\Sset|}{(|\Sset| - 1) \log n} \leq \frac{2}{\log n} \text{,}
  \end{equation*}
  which completes the proof.
\end{proof}

\begin{lemma}\label{lem:redundancysum}
  The sum of the redundancy of $\sigma$ RRR bitvectors of $m_1, \dots, m_\sigma$
  bits respectively, where $\sum_i m_i = m$, can be bounded by
  \begin{equation*}
    O\left(m \frac{\log \log \frac{m}{\sigma}}{\log \frac{m}{\sigma}} + \sigma\right) \text{.}
  \end{equation*}
\end{lemma}
\begin{proof}
  The redundancy of a single bitvector can be bounded by $c_1
  \frac{m_i \log \log m_i}{\log m_i} + c_2$. Since $f(x) = \frac{x
    \log\log x}{\log x}$ is concave, we can apply the Jensen
  inequality:
  \begin{equation*}
    \frac{1}{\sigma}\sum_{i} \left( c_1 \frac{m_i \log \log m_i}{\log m_i} + c_2 \right) \leq c_1 \frac{\frac{m}{\sigma} \log \log \frac{m}{\sigma}}{\log \frac{m}{\sigma}} + c_2 \text{.}
  \end{equation*}

  The result follows by multiplying both sides by $\sigma$.
\end{proof}

\begin{lemma}\label{lem:redundancywt}
  The redundancy of the RRR bitvectors in $WT(S)$ can be bounded by $o(\tilde h n)$.
\end{lemma}
\begin{proof}
  Since the bitvector lengths add up to $\tilde h n$, we can apply
  Lemma~\ref{lem:redundancysum} and obtain that the redundancy are
  bounded by
  \begin{equation*}
    O\left(\tilde h n \frac{\log \log \frac{\tilde h n}{|\Sset|}}{\log \frac{\tilde h n}{|\Sset|}} + |\Sset|\right) \text{.}
  \end{equation*}
  The term in $|\Sset|$ is already taken care of by
  Lemma~\ref{lem:sigmanegligible}. It suffices then to prove that
  \begin{equation*}
    \frac{\log \log \frac{\tilde h n}{|\Sset|}}{\log \frac{\tilde h n}{|\Sset|}}
  \end{equation*}
  is negligible as $n$ grows, and because $f(x) = \frac{\log \log
    x}{\log x}$ is asymptotic to $0$, we just need to prove that
  $\frac{\tilde h n}{|\Sset|}$ grows to infinity as $n$ does. Using
  again Lemma~\ref{lem:hbound} and Lemma~\ref{lem:entropymin} we
  obtain that
  \begin{equation*}
    \frac{\tilde h n}{|\Sset|} \geq \frac{nH_0(S)}{|\Sset|} \geq \frac{(|\Sset| - 1) \log n}{|\Sset|} \geq \frac{\log n}{2}
  \end{equation*}
  thus proving the lemma.
\end{proof}

\begin{lemma}\label{lem:rrrdelim}
  The partial sum data structure used to delimit the RRR bitvectors in
  $WT(S)$ occupies $o(\tilde h n)$ bits.
\end{lemma}
\begin{proof}
  By Lemma~\ref{lem:redundancywt} the sum of the RRR encodings is $n
  H_0(S) + o(\tilde h n)$. To encode the $|\Sset|$ delimiters, the
  partial sum structure of \cite{RRR07} takes
  \begin{equation*}
    \begin{split}
      &|\Sset|\log\left(\frac{nH_0(S) + o(\tilde h n) +
          |\Sset|}{|\Sset|}\right) + O(|\Sset|) \\
      & \hspace*{1.2em} \leq |\Sset| \log\left(\frac{nH_0(S)}{|\Sset|}\right) +
      |\Sset|\log\left(\frac{o(\tilde h n)}{|\Sset|}\right) +
      O(|\Sset|) \text{.}
    \end{split}
  \end{equation*}
  The third term is negligible by Lemma~\ref{lem:sigmanegligible}. The
  second just by dividing by $\tilde h n$ and noting that $f(x) =
  \frac{\log x}{x}$ is asymptotic to $0$. It remains to show that the
  first term is $o(\tilde h n)$. Dividing by $\tilde h n$ and using
  again Lemma~\ref{lem:hbound} we obtain
  \begin{equation*}
    \frac{|\Sset| \log\left(\frac{nH_0(S)}{|\Sset|}\right)}{\tilde h n} \leq \frac{|\Sset| \log\left(\frac{nH_0(S)}{|\Sset|}\right)}{n H_0(S)} = \frac{\log\left(\frac{nH_0(S)}{|\Sset|}\right)}{\frac{n H_0(S)}{|\Sset|}}
  \end{equation*}
  By using again that $f(x) = \frac{\log x}{x}$ is asymptotic to $0$
  and proving as in Lemma~\ref{lem:redundancywt} that
  $\frac{nH_0(S)}{|\Sset|}$ grows to infinity as $n$ does, the result
  follows.
\end{proof}

\section{Dynamic Patricia Trie}\label{sec:patriciaops}

For the dynamic Wavelet Trie we use a straightforward Patricia Trie
data structure. Each node contains two pointers to the children, one
pointer to the label and one integer for its length. For $k$ strings,
this amounts to $O(k w)$ space. Given this representation, all
navigational operations are trivial. The total space is $O(kw) + |L|$,
where $L$ is the concatenation of the labels in the compacted trie as
defined in Theorem~\ref{thm:ltbound}

Insertion of a new string $s$ \emph{splits} an existing node, where
the mismatch occurs, and adds a leaf. The label of the new internal
node is set to point to the label of the split node, with the new
label length (corresponding to the mismatch of $s$ in the split
node). The split node is modified accordingly. A new label is
allocated with the suffix of $s$ starting from the mismatch, and
assigned to the new leaf node. This operation takes $O(|s|)$ time and
the space grows by $O(w)$ plus the length of the new suffix, hence
maintaining the space invariant.

When a new string is deleted, its leaf is deleted and the parent node
and the other child of the parent need to be merged. The highest node
that shares the label with the deleted leaf is found, and the label is
deleted and replaced with a new string that is the concatenation of
the labels from that node up to the merged node. The pointers in the
path from the found node and the merged node are replaced accordingly.
This operation takes $O(\hat \ell)$ where $\hat \ell$ is the length of
the longest string in the trie, and the space invariant is maintained.

\end{document}